\documentclass[10pt]{article}

\usepackage{amsmath}
\usepackage{amsthm}
\usepackage{amssymb}
\usepackage{amscd}
\usepackage{fancybox,ifthen,float,epsfig,subfigure}
\usepackage{graphicx}
\usepackage{dcolumn}
\usepackage{bm}
\usepackage{soul}
\usepackage[margin=1 in]{geometry}

\newcommand\pa{\partial}\newcommand\w{\wedge}
\newcommand\tot{\operatorname{tot}}
\newcommand\In{\operatorname{in}}

\newcommand\dR{\operatorname{dR}}

\newcommand\bA{\boldsymbol A}

\newcommand\bM{\boldsymbol M}

\newcommand\bH{\boldsymbol H}

\newcommand\bE{\boldsymbol E}

\newcommand\bJ{\boldsymbol J}

\newcommand\bx{\boldsymbol x}
\newcommand\by{\boldsymbol y}

\newcommand\bn{\boldsymbol n}
\newcommand\br{\boldsymbol r}
\newcommand\ds{d{\boldsymbol s}}
\newtheorem{theorem}{Theorem}
\newtheorem{remark}{Remark}

\begin{document}

\title{A consistency condition for the vector potential in multiply-connected 
domains}

\author{Charles L. Epstein\thanks{Departments of Mathematics
    and Radiology, University of Pennsylvania,
    209 South 33rd St., Philadelphia, PA 19104.
    E-mail: {cle@math.upenn.edu}.
    Research partially supported by NSF grants DMS06-03973 and 
    DMS09-35165, 
    DARPA grant HR0011-09-1-0055, 
    and NSSEFF Program Award FA9550-10-1-0180.}, \
Zydrunas Gimbutas\thanks{Courant
    Institute,
    New York University, 251 Mercer Street, New York, NY 10012.
    E-mail: {gimbutas@cims.nyu.edu}. 
    Research partially supported by the Air Force Office of 
    Scientific Research under MURI grant FA9550-06-1-0337 and 
    NSSEFF Program Award FA9550-10-1-0180.}, \
Leslie Greengard\thanks{Courant Institute,
    New York University, 251 Mercer Street, New York, NY 10012.
    E-mail: {greengard@cims.nyu.edu}. 
    Research partially supported by the U.S. Department of Energy under 
    contract DEFG0288ER25053 and by the Air Force Office of Scientific
    Research under MURI grant FA9550-06-1-0337 and NSSEFF Program 
    Award FA9550-10-1-0180.}, \\
Andreas Kl\"ockner\thanks{Courant Institute,
    New York University, 251 Mercer Street, New York, NY 10012.
    E-mail: {kloeckner@cims.nyu.edu}. 
    Research partially supported by the U.S. Department of Energy 
    under contract DEFGO288ER25053 and by the Air Force Office of Scientific 
    Research under NSSEFF Program Award FA9550-10-1-0180.}, \
and Michael O'Neil\thanks{Courant Institute,
    New York University, 251 Mercer Street, New York, NY 10012.
    E-mail: {oneil@cims.nyu.edu}. 
    Research partially supported by the Air Force Office of Scientific
    Research under NSSEFF Program Award FA9550-10-1-0180.
    }}

\date{\today}

\maketitle
\numberwithin{equation}{section}

\begin{abstract}
  A classical problem in electromagnetics concerns the representation
  of the electric and magnetic fields in the low-frequency or static
  regime, where topology plays a fundamental role. For multiply
  connected conductors, at zero frequency the standard boundary
  conditions on the tangential components of the magnetic field do not
  uniquely determine the vector potential.  We describe a
  (gauge-invariant) consistency condition that overcomes this
  non-uniqueness and resolves a longstanding difficulty in inverting
  the magnetic field integral equation.
\end{abstract}

\section{\label{sec:intro}Introduction}

We consider the problem of exterior electromagnetic 
scattering in the frequency domain, with particular attention paid to 
the behavior of the electric and magnetic fields in the static limit.
Following standard practice, we write 
$\bE^{\tot}(\bx) = \bE^{\In}(\bx) + \bE(\bx)$ and
$\bH^{\tot}(\bx) = \bH^{\In}(\bx) + \bH(\bx)$,
where $\{ \bE^{\In}, \bH^{\In} \}$ describe known incident electric and magnetic 
fields, $\{ \bE, \bH \}$ denote the scattered field of interest,
and $\{ \bE^{\tot}, \bH^{\tot} \}$ denote the {\em total} fields.
We write Maxwell's equations in the form
\begin{eqnarray}
\label{maxwellfd}
\nabla \times \bH^{\tot} &=& - i \omega \epsilon \bE^{\tot}  \\
\nabla \times \bE^{\tot} &=&  i \omega \mu \bH^{\tot} \, . \nonumber
\end{eqnarray}
where $\epsilon, \mu$ are
the permittivity and permeability, and 
we define the wavenumber by $k = \omega \sqrt{\epsilon \mu}$.
The scattered field is assumed to satisfy the 
Sommerfeld-Silver-M\"uller radiation condition:
\begin{equation}
\label{silvermuller}
 \lim_{r\rightarrow \infty}
 \left( \bH \times \frac{\br}{r} - 
\frac{\mu}{\epsilon} \bE \right) = o \left( \frac{1}{r} \right).
\end{equation}

For a perfect conductor $\Omega$,
two homogeneous conditions to be enforced on $\Gamma$, the 
boundary of $\Omega$, are \cite{JACKSON,PAPAS}:
\begin{equation}
\bn \times \bE^{\tot} = 0, \quad  \bn \cdot \bH^{\tot} = 0.
\label{EHbc}
\end{equation}
It is also well-known that, on $\Gamma$, we must have
\begin{equation}
\bn \times \bH^{\tot} = \bJ, \quad  \bn \cdot \bE^{\tot} = \rho/\epsilon \, ,
\label{EHbc2}
\end{equation}
assuming that the scattered fields are induced by a physical surface 
current $\bJ$ and a corresponding charge $\rho$ with
\begin{equation}
\nabla_\Gamma \cdot \bJ = i\omega \rho \, ,
\label{contcondition}
\end{equation}
where $\nabla_\Gamma$ denotes the surface divergence.

Our primary interest  here is with the classical representation of
electromagnetic fields in terms of the vector and scalar potential 
(in the Lorenz gauge):
\begin{eqnarray}
\bE &=& i \omega \bA - \nabla \phi  \label{Epotrep} \\
\bH &=& \frac{1}{\mu} \nabla \times \bA \label{Hpotrep} 
\end{eqnarray}
where 
$\bA[\bJ](\bx) = \mu \int_\Gamma g_k(\bx-\by) \bJ(\by) dA_{\by}$
and $\phi(\bx) = \frac{1}{i\omega\epsilon\mu} \nabla \cdot \bA$,
with
$g_k(\bx) = \frac{e^{ik |\bx|}}{ 4 \pi | \bx|}$.
In scattering problems,
$\bJ$, the unknown surface current, is a tangential vector field.

Using the first condition in (\ref{EHbc2}), with care in taking the
limit as $\bx$ approaches the boundary $\Gamma$ from the exterior
domain, we obtain the magnetic field integral equation (MFIE):
\begin{equation}
\frac{1}{2} \bJ(\bx) - \bM_k[\bJ](\bx)
= \bn(\bx) \times \bH^{\In}(\bx) \qquad  (\bx \in \Gamma),
\label{mfie}
\end{equation}
where 
\begin{equation}
\bM_k[\bJ](\bx) = \bn(\bx) \times \nabla \times 
\int_\Gamma g_k(\bx-\by) \bJ(\by) dA_{\by} \, .
\label{Kdef}
\end{equation}
It is a second-kind Fredholm equation,
originally suggested by Maue \cite{MAUE}.
See \cite{CK,CHEW,JIN,NEDELEC} for detailed discussions.

While the MFIE has a sequence of spurious resonances at higher and
higher frequencies, below the first such resonance it yields an
invertible and well-conditioned linear system, so long as the
scatterer is simply-connected.  When the scatterer is topologically
non-trivial, however, the MFIE breaks down, having a non-trivial
null-space at zero frequency.  In the static limit, the dimension of
the null-space is equal to the genus of the surface. While this
problem has been carefully analyzed \cite{MICH,CK,WERNER}, no simple
remedy has been provided to date.

In this paper, we show that the difficulties encountered by
the MFIE can be resolved
through the enforcement of an apparently new consistency condition on the 
vector potential, involving line integrals around \emph{$B$-cycles} (see Fig. \ref{fig1}). 
Like the Aharonov-Bohm effect, it illustrates the fundamental role of the vector
potential, but in a classical regime.
We begin with the static case, since it is mathematically simpler and of importance 
in its own right.

\begin{figure}
\centering
\includegraphics[width=8cm]{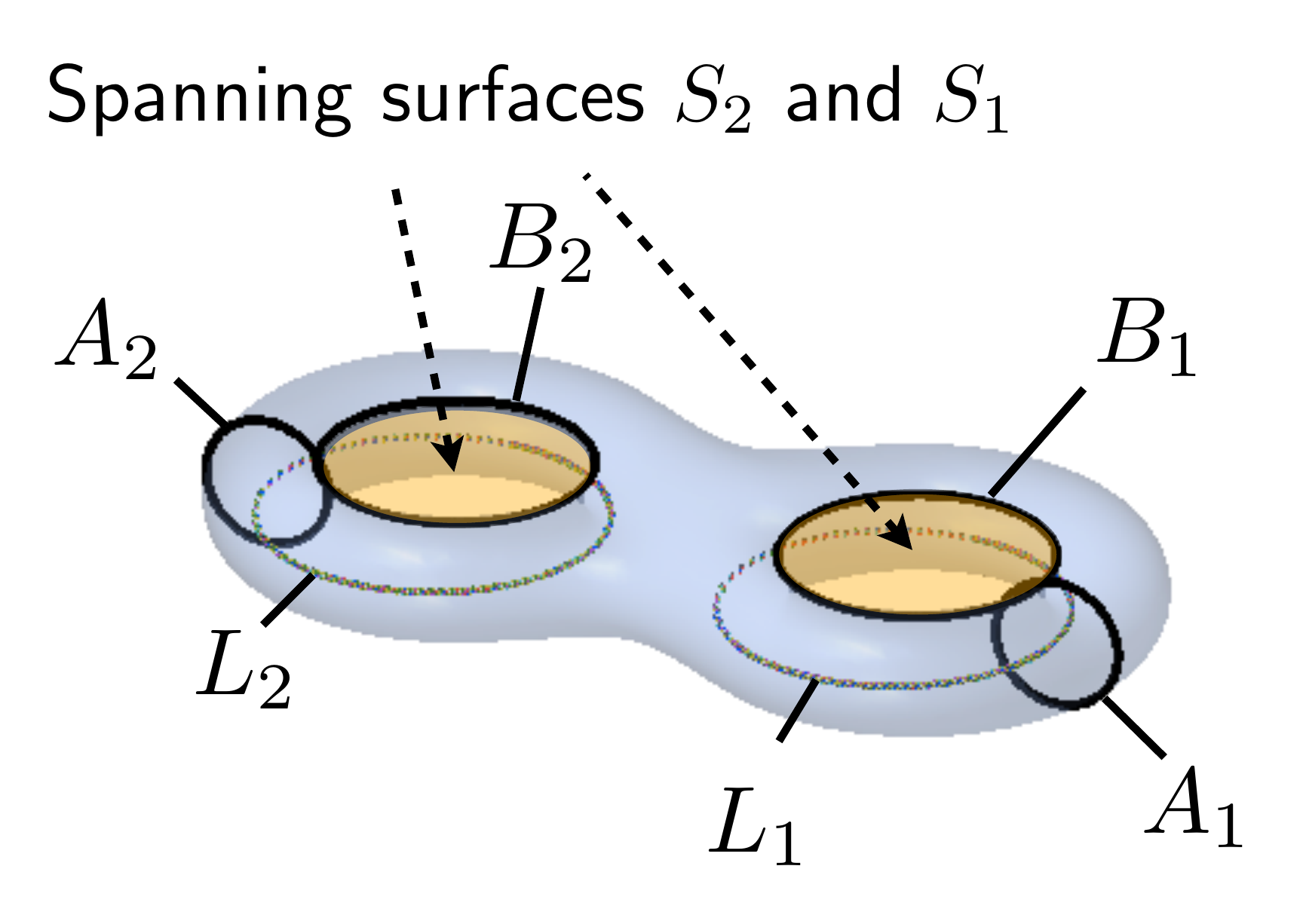}
\caption{A surface of genus 2. The $A$-cycles are loops whose spanning 
surface is in the interior of the scatterer. The $B$-cycles have 
spanning surfaces $S_1,S_2$ 
which lie in the 
exterior domain. The loops $L_1,L_2$ lie in the interior domain and pass through 
the $A$-cycles.}
\label{fig1}
\end{figure}

\section{Magnetostatics}

In the zero frequency limit, Maxwell's equations are generally said to uncouple,
with \emph{scattered} electrostatic and magnetostatic fields denoted by ($\bE_0, \bH_0$),
respectively. In the unbounded domain $\mathbb{R}^3\setminus\overline{\Omega}$ 
exterior to $\Omega$,
$\bE_0$ satisfies 
\begin{equation}
\nabla \times \bE_0 = 0, \  
\nabla \cdot \bE_0 = 0, 
\label{E0eqs}
\end{equation}
subject to the boundary condition
\begin{equation}
\bn \times \bE_0 = - \bn \times \bE_0^{\In}\big|_\Gamma \, ,
\label{Estatbc}
\end{equation}
where $\bn$ denotes the outward normal to $\Omega$. 
In $\mathbb{R}^3\setminus\overline{\Omega}$, $\bH_0$ satisfies
\begin{equation}
\nabla \times \bH_0 = 0, \  
\nabla \cdot \bH_0 = 0, 
\label{H0eqs}
\end{equation}
subject to the boundary condition
\begin{equation}
\bn \cdot \bH_0 = -\bn \cdot \bH_0^{\In} \big|_\Gamma \, .
\label{Hstatbc}
\end{equation}
Vector fields that are both curl-free and divergence-free are called
harmonic vector fields.

A fundamental difficulty arises in the static theory that is topological
in nature. More precisely, if the number of connected components of $\Omega$
is denoted by $m$, then there is an $m$-dimensional space of nontrivial
vector fields $\bE_0^h$ 
in $\mathbb{R}^3\setminus\overline{\Omega}$, satisfying
\begin{equation}
\nabla \times \bE_0^h = 0, \  \nabla \cdot \bE_0^h = 0, \ 
\bn \times \bE_0^h = 0 \big|_\Gamma \, .
\end{equation}
They are generally referred to as (exterior) Dirichlet fields and
correspond to setting the electrostatic potential $\phi$ on each
disjoint conductor to a different constant.
This space of fields is, of course, essential in studying
capacitance problems in a system of disjoint conductors, 
where $\bE_0^h = -\nabla \phi$.

Here, we concentrate on magnetostatics, where nontrivial solutions
$\bH_0(\bx)$ to 
\begin{equation}
\nabla \times \bH_0 = 0, \  \nabla \cdot \bH_0 = 0, \ 
\bn \cdot \bH_0 = 0 \big|_\Gamma \, ,
\end{equation}
are referred to as interior or exterior Neumann fields, depending on whether 
$\bx \in \Omega$ or $\bx \in \mathbb{R}^3\setminus\overline{\Omega}$, respectively. 
We let $H^+$ denote the space of exterior Neumann fields and 
$H^-$ denote the space of interior Neumann fields.
It is a classical fact that the dimension
of the spaces $H^+$ and $H^-$ is $g$, 
where $g$ is the genus of the boundary $\Gamma$
\cite{CK}.

In short, the boundary condition  (\ref{Hstatbc}) alone determines a unique field
only in the simply connected case ($g=0$). A natural representation 
in that setting is to seek $\bH_0$ as the gradient of the magnetic scalar potential,
$\bH_0 = \nabla \Psi$, with $\Psi$ a single-valued harmonic function satisfying
\begin{equation}
 \bn \cdot \nabla \Psi = -\bn \cdot \bH_0^{\In} \, .
\end{equation}

In the multiply connected case,
additional data is needed to make the magnetostatic problem well-posed,
all of which are designed to account for current loops that may be flowing
through the handles of the domain. One such condition is to require that
the line integrals of $\bH_0$ around each $A$-cycle be specified:
\begin{equation}
 \int_{A_j} \bH_0 \cdot \ds = \alpha_i \, .  \label{Hextra}
\end{equation}
It can be shown that the solution to 
(\ref{H0eqs},\ref{Hstatbc},\ref{Hextra}) is unique (see, for example,
\cite{VB,KRESS,CK,BOSSAVIT,MAYERGOYZ,HENROTTE}).

The use of the scalar potential can be extended to the multiply
connected case in two ways - either by introducing $g$ cuts on the
boundary $\Gamma$ and allowing for a potential jump across each cut
\cite{VOURDAS}, or by introducing $g$ current loops in the interior domain
$\Omega$ that span the fundamental group of $\Omega$ (essentially one
passing through each $A$-cycle - see Fig. \ref{fig1}) \cite{VB,BOSSAVIT}. 
In the latter case, it is straightforward to see that one can represent
$\bH_0$ as
\[ 
\bH_0 = \nabla \Psi + \frac{1}{\mu} \nabla \times \sum_{i=1}^g \alpha_i \bA[{\cal L}_i] \, ,
\]
where ${\cal L}_i$ is a loop of unit current density flowing through $L_i$
(Fig. \ref{fig1}), allowing the scalar potential $\Psi$ to be single-valued.

It is also possible to extend the scalar potential approach
to the full Maxwell equations at non-zero frequencies through a 
generalization of the Lorenz-Debye-Mie formalism given in \cite{EG,EGO}, 
but this involves non-physical variables. Unlike the above formulation,
the MFIE (\ref{mfie}) uses a physical unknown, the surface 
current $\bJ,$ and extends naturally away from the static limit through the
representation (\ref{Epotrep})-(\ref{Hpotrep}).  In this paper we seek
to better understand, in the static case, the range and null-space of the MFIE. 

\section{Magnetostatics using the vector potential}
   
There is a substantial literature on magnetostatics and the representation of 
harmonic vector fields in the form of a curl (see, for example,
\cite{MICH,WERNER,BOSSAVIT,HENROTTE,FRIEDRICHS,VOURDAS}). We do not seek to review 
the theory here, except where it is of direct relevance to the MFIE.
At zero frequency, the MFIE takes the form
\begin{equation}
\left[ \frac{1}{2} I - M_0 \right] \bJ
= \bn \times \bH^{\In} \, ,
\label{mfie0}
\end{equation}
where $I$ denotes the identity operator and we have dropped the
argument $\bx$ for the sake of clarity.  We let
$\{H^+_j:\:j=1,\dots,g\}$ denote a basis for the exterior Neumann
fields and $\{H^-_j:\:j=1,\dots,g\}$ denote a basis for the interior
Neumann fields.  We let $\{Z^+_j = H^+_j \big|_\Gamma:\:j=1,\dots,g\}$
denote the boundary values of the exterior Neumann fields and $\{Z^-_j
= H^-_j \big|_\Gamma:\:j=1,\dots,g\}$ denote the boundary values of
interior Neumann fields, which are perforce vector fields tangent to
$\Gamma.$ As shown in \cite{WERNER,CK,MICH}, the relevant null-spaces are
known to be
\begin{equation}
  \begin{split}
    N \left(\frac{1}{2}I - M_0 \right) &=\{\bn\times Z^+_j:\:j=1,\dots,g\}\\
 N \left(\frac{1}{2}I - M_0' \right) &=\{Z^-_j:\:j=1,\dots,g\} \, ,
  \end{split}
\end{equation}
where $M'$ is the adjoint of $M$.  From Fredholm theory, the
solvability conditions for the MFIE are, therefore, that the inner
products $(\bn\times\bH^{\In},Z^{-}_j)=0$ for $j=1,\dots, g.$ This is
always the case for an incoming field that is generated by currents lying
outside of a simply connected neighborhood of $\Omega$ \cite{WERNER}.
In short, for the physical scattering problem, the operator
$\frac{1}{2}I - M_0$ is {\em rank} deficient but not {\em range}
deficient, and it remains only to develop a set of constraints that
make it invertible.

\begin{figure}
\centering
\includegraphics[width=15cm]{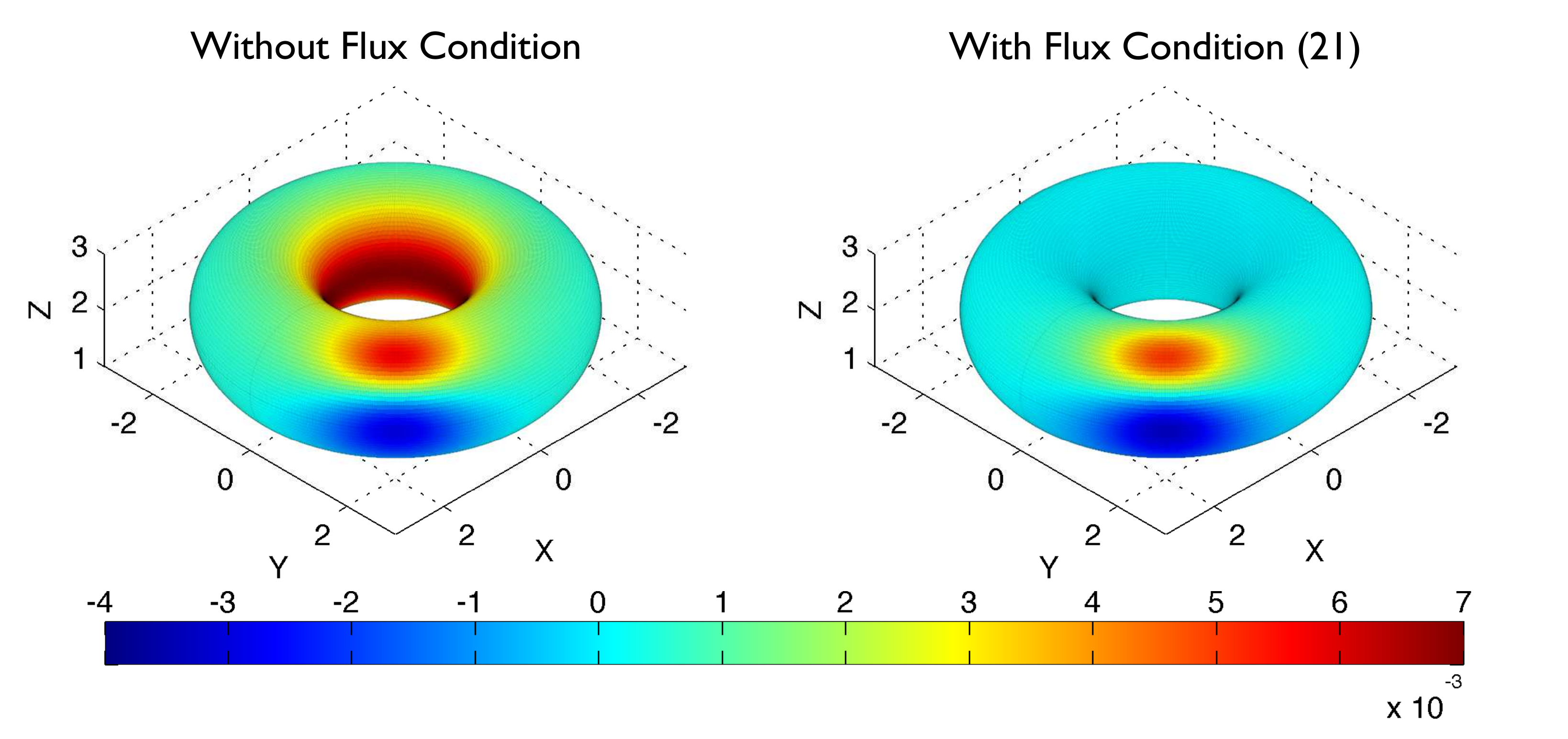}
\caption{The induced surface current on a torus
at wavenumber $k=10^{-16}$ without (left) and with (right) the auxiliary 
consistency (flux) condition. The real part of the current in the azimuthal
direction is shown. 
The incident field is due to a unit strength current 
loop of radius $.5$ located at 
$(3,3,4)$.}
\label{fig2}
\end{figure}

\begin{figure}
\centering
\includegraphics[width=15cm]{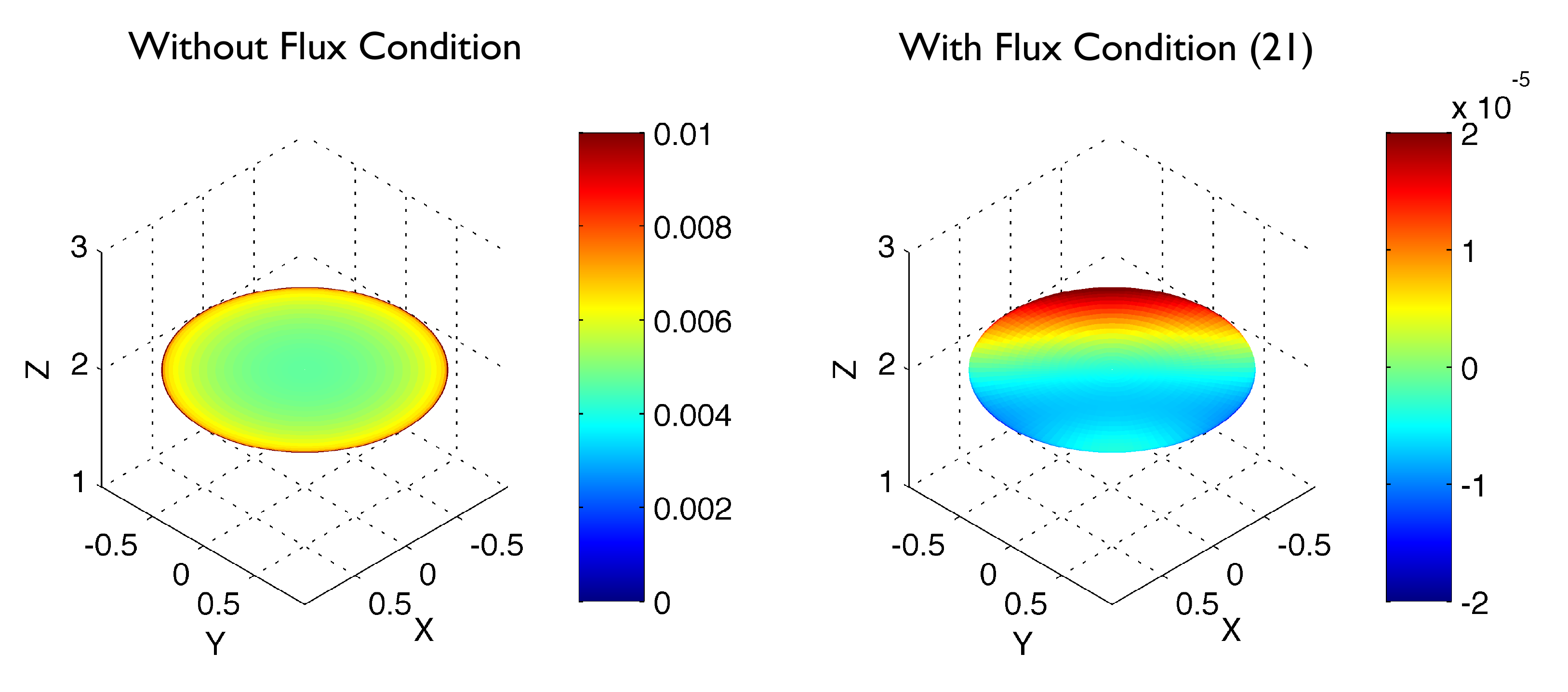}
\caption{The real part of $\bH^{\tot} \cdot \hat{\boldsymbol{z}}$ 
on surfaces spanning
the holes of the tori in Fig. \ref{fig2}. Note the different scales on the 
color bars, as well as the change in sign of the point-wise flux
on the right. The net flux through the left surface is $0.019$ while,
on the right, it is zero to machine precision.}
\label{fig3}
\end{figure}

We turn now to the main point of the present paper: that such constraints
can be found and that they come from an analysis of the {\em electric} field
in the limit $\omega \rightarrow 0$. 
We begin by noting that on a perfect conductor,
the vanishing of the total tangential electric field (\ref{EHbc}) and
Stokes' theorem allows us to write 
\[
 \int_{B_j} \bE \cdot {\ds} =  - \int_{B_j} \bE^{\In} \cdot {\ds} =
- \int_{S_j} i \omega \mu \, \bH^{\In} \cdot {\bn} \, dA \, ,
\]
where $S_j$ is a spanning surface for the $B$-cycle $B_j$.
Dividing both sides by $i\omega$, using the 
representation (\ref{Epotrep}), and noting that the line integral
of a gradient vanishes, we have

\begin{theorem} \label{thm1}
Let $\Omega$ be a multiply-connected perfect conductor and let $B_j$ be a 
$B$-cycle and  $S_j$ be a spanning surface for $B_j$. Then
\begin{equation}
 \int_{B_j} \bA \cdot \ds = - \mu \int_{S_j} \bH^{\In} \cdot \bn dA \, . 
\label{Acond1}
\end{equation}
\end{theorem}

This condition makes sense for any $\omega \geq 0$. 
Alternatively, if we have access to the vector potential 
$\bA^{\In}$ which defines the incoming field, the same analysis yields the
even simpler condition:
\begin{equation}
\int_{B_j} \bA \cdot \ds = - \int_{B_j} \bA^{\In} \cdot \ds.
\label{Acond2}
\end{equation}

\begin{theorem}
Let $\Omega$ be a multiply-connected perfect conductor of genus $g$
with $B$-cycles $B_1,\dots,B_g$.
Then the MFIE augmented by conditions of the form (\ref{Acond1})
or (\ref{Acond2}) for $j=1,\dots,g$ has a unique solution.
\end{theorem}
\begin{proof} As noted above, the equation $(\frac 12
  I-M_0)\bJ=\bn\times\bH_0^{\In}$ is solvable for any physically
  meaningful right hand side. Thus the only issue is that of
  uniqueness; we need to show that the combined null-space of
  the operator $(\frac 12 I-M_0)$ and~\eqref{Acond1} is trivial. The
  null-space of the integral operator is spanned by $\{\bn\times Z_j^+:
  j=1,\dots,g\}.$ Hence, if $\bJ$ lies in the null-space of this
  integral operator, then the exterior field takes the form
  \begin{equation}
    \bH= \frac{1}{\mu} \nabla\times \bA[\bJ]= \sum_{j=1}^g\beta_jH_j^+,
  \end{equation}
see \cite{WERNER}. Under the correspondence between vector fields and
  2-forms
\begin{equation}
  h_1\pa_{x_1}+h_2\pa_{x_2}+h_3\pa_{x_3} \leftrightarrow
h_1 dx_2\w dx_3+h_2dx_3\w dx_1+h_3 dx_1\w dx_2,
\end{equation}
the fields $\{H_j^+: j=1,\dots, g\}$ are a basis for the relative cohomology group 
$H^2_{\dR}(\mathbb{R}^3\setminus\overline{\Omega},\Gamma).$  The functionals
\begin{equation}
  \bH\mapsto \int\limits_{S_j}\bH\cdot\bn dA= \frac{1}{\mu} \int\limits_{B_j}\bA[\bJ]\cdot \ds
\text{ for }j=1,\dots,g
\end{equation}
span the dual space to this vector space. Hence if these integrals all
vanish, then the coefficients $\{\beta_j\}$ must all vanish as well.
\end{proof}

\begin{remark} 
Let us suppose now that we have discretized the MFIE (\ref{mfie}) with
$2N$ unknowns (2 degrees of freedom for the surface current $\bJ$ at each
of $N$ points), resulting in the linear system
$A \boldsymbol j = \boldsymbol b$ and the constraints (\ref{Acond2}) by 
$C \boldsymbol j = \boldsymbol f$,
where $C$ is a $g \times 2N$ matrix.
It is straightforward to show that the system  
\[ [ A + QC] \boldsymbol j = \boldsymbol b + Q\boldsymbol f \]
has the same solution as the constrained equation. It is invertible so
long as the range of $Q$, a $2N \times g$ matrix, has a full rank
projection onto the null vectors of the adjoint $A'$.  In the low frequency
regime, it is sufficient to use $Q = C'$.
\end{remark}

\section{Time harmonic electromagnetics using the vector potential}

As soon as $\omega \neq 0$, the MFIE (\ref{mfie}) is formally
invertible and there is no need for the incorporation of consistency
constraints. Because there is a nearby singular problem, however, the
linear system is extremely ill-conditioned at low frequency. Without
the additional constraints the {\em condition number} is
$O(\omega^{-2}).$ In that regime, the incorporation of the constraints
improves the condition number considerably.

\section{Numerical Results}

For illustration, we consider the problem of scattering from a torus:
a genus one surface of revolution (Fig. \ref{fig2}), driven by a known
current loop in the exterior.
We have implemented a solver for the MFIE with and without the consistency condition 
(with a detailed discussion of the method to be reported
at a later date). One convenient measure of the error in the solution 
is the flux of $\bH^{\tot}$ through the hole in the torus,
which should vanish on a perfect conductor
(see Fig. \ref{fig3}).

\section{Conclusions}

In this paper, we have derived a simple consistency
condition (Theorem 1) for the vector potential in the context of scattering 
from perfect conductors, valid at any frequency.
It is of physical interest for three 
reasons: 1) it describes an interesting correlation between the electric and 
magnetic field that persists at zero frequency, 
2) it enforces uniqueness for the magnetic field integral 
equation (MFIE) in  multiply connected domains in the static limit, and 3)
it improves the stability and robustness of the MFIE in the low 
frequency regime.

It is, perhaps, worth noting that the role of the vector potential in
the consistency condition is, in some sense,
dual to its role in the Aharonov-Bohm effect. In the latter case,
it is the line integral of the vector potential around $A$-cycles
that is critical.
That integral measures the {\em solenoidal} ({\em poloidal}) current flow 
on the torus, which induces a zero electromagnetic field in the exterior. 
The line integral of $\bA$ around $B$-cycles, on the other hand, is sensitive to
the {\em toroidal} current flow on the surface. On a perfect conductor,
that induced current exactly cancels the flux of the incoming magnetic
field through the ``holes.''

\section{Acknowledgments}
We wish to thank Eric Michielssen for several useful discussions.

\end{document}